\let\accentvec\vec 

\documentclass[11pt]{article} 
\usepackage[margin=1.4in]{geometry}
\usepackage[english]{babel}
\usepackage[T1]{fontenc}
\usepackage{lmodern} 
\usepackage{graphicx}
\usepackage{wrapfig}
\usepackage{subfig}
\usepackage{paralist}

\usepackage{etoolbox}
\newtoggle{long}
\toggletrue{long}

\newcommand{\titel}{Small-Area Orthogonal Drawings \\ of 3-Connected Graphs}

\usepackage{algorithm} 
\usepackage[noend]{algpseudocode} 

\usepackage[usenames]{color} 
\definecolor{hellblau}{rgb}{0.2,0.4,1} 
\definecolor{dunkelblau}{rgb}{0,0,0.8}
\definecolor{dunkelgruen}{rgb}{0,0.5,0}

\usepackage{xspace}
\usepackage{pdfpages}

 \let\vec\accentvec 
    
\usepackage{amsmath} 
\usepackage{amsthm} 
\usepackage{amsfonts} 
\theoremstyle{plain} 
	\newtheorem{satz}{Satz}[] 
	\newtheorem{theorem}[satz]{Theorem}
	\newtheorem{lemma}[satz]{Lemma}
	\newtheorem{observation}[satz]{Observation}
	\newtheorem{proposition}[satz]{Proposition}
\theoremstyle{remark} 
	\newtheorem*{remark}{\textbf{Remark}} 
\theoremstyle{definition} 
	\newtheorem{definition}[satz]{Definition}
	\newtheorem{corollary}[satz]{Corollary}
	\newtheorem*{conjecture}{Conjecture}
	\newtheorem{example}[satz]{Example}

\newcommand{\x}[2]{\ensuremath{x^{#1{\mbox{-}#2}}}}

\newcommand{\indeg}{\mathit{indeg}}
\newcommand{\outdeg}{\mathit{outdeg}}

\renewcommand{\paragraph}[1]{\smallskip\noindent{\em #1}}

\begin{document}
	\title{\titel}
	\author{}
		\author{Therese Biedl\thanks{David R. Cheriton School of Computer Science, University of Waterloo, Canada.  Supported by NSERC.} \and Jens M. Schmidt\thanks{Institute of Mathematics, TU Ilmenau, Germany.}}
	\date{}
	\maketitle

\begin{abstract}
It is well-known that every graph with maximum degree 4 has an orthogonal
drawing with area at most $\frac{49}{64}n^2+O(n)\approx 0.76n^2$. In this paper, we show that if the
graph is 3-connected, then the area can be reduced even further to $\frac{9}{16}n^2+O(n) \approx 0.56n^2$. The drawing
uses the 3-canonical order for (not necessarily planar) 3-connected graphs, which can be computed
in linear time from the Mondshein-sequence. To our knowledge, this is the first application of the 3-canonical
order on non-planar graphs in graph-drawing.
\end{abstract}


\section{Introduction}

An orthogonal drawing of a graph $G=(V,E)$ is an assignment of vertices to \emph{points} and edges to \emph{polygonal lines}
connecting their endpoints such that all edge-segments are horizontal or vertical. Edges
are allowed to intersect, but only in single points that are not bends of the polygonal lines.
Such an orthogonal drawing can exist only if every vertex has degree at most 4;
we call such a graph a {\em 4-graph}. It is easy to see that every 4-graph has an orthogonal drawing with area $O(n^2)$,
and this is asymptotically optimal \cite{Valiant1981}.

For planar 2-connected graphs, multiple
authors showed independently \cite{TT86,RT86} how to achieve area $n\times n$,
and this is optimal \cite{TTV91}.
We measure the drawing-size as follows. Assume (as we do throughout the paper) 
that all vertices and bends are at points with integral coordinates.
If $H$ rows and $W$ columns of the integer grid intersect the drawing, then we say that
the drawing occupies a {\em $W\times H$-grid} with {\em width} $W$, {\em height} $H$, {\em half-perimeter}
$H+W$ and {\em area} $H\cdot W$.%
\footnote{Some papers count as width/height the width/height of the smallest enclosing axis-aligned
box. This is one unit less than with our measure.}

For arbitrary graphs (i.e., graphs that are not necessarily planar), improved bounds on the area of
orthogonal drawings were developed much later, decreasing from $4n^2$ \cite{Schaeffter95} to $n^2$
\cite{Biedl1998} to $0.76n^2$ \cite{Papakostas1998}.  (In all these statements we omit lower-order terms
for ease of notation.)

\paragraph{Our results: }  In this paper, we decrease the area-bound for orthogonal drawings
further to $0.56n^2+O(n)$ under the assumption that the graph is 3-connected.  The approach
is similar to the one by Papakostas and Tollis \cite{Papakostas1998}: add vertices to the drawing
in a specific order, and pair some of these vertices so that in each pair one vertex re-uses
a row or column that was used by the other. The main difference in our paper is that 
3-connectivity allows the use of a different, stronger, vertex order.

It has been known for a long
time that any {\em planar} 3-connected graph has a so-called canonical order \cite{FPP90,Kant1996}, 
which is useful for planar graph drawing algorithms. It was mentioned
that such a canonical order also exists in non-planar graphs (e.g.~in \cite[Remark on p.113]{FM94}), but
it was not clear how to find it efficiently, and it has to our knowledge not been used
for graph drawing algorithms.  Recently, the second author studied the so-called
Mondshein-sequence, which is an edge partition of a 3-connected graph with special
properties \cite{Mondshein1971}, and showed that it can be computed in linear time \cite{Schmidt2014}.  
From this Mondshein-sequence, one can easily find the canonical order for non-planar 3-connected graphs
in linear time~\cite{Schmidt2013}; we call this a {\em 3-canonical order}.

We use this 3-canonical order to add vertices to the orthogonal drawing. This almost
immediately lowers the resulting area, because vertices with one incoming edge can 
only occur in chains.  We then mimic the pairing-technique of Papakostas and Tollis,
and pair groups of the 3-canonical order in such a way that even more rows and columns
can be saved, resulting in a half-perimeter of $\frac{3}{2}n+O(1)$ and the area-bound
follows.

No previous algorithms were known that achieve smaller area for 3-connected 4-graphs 
than for 2-connected 4-graphs. For {\em planar} graphs, the orthogonal drawing
algorithm by Kant \cite{Kant1996} draws 
3-connected planar 4-graphs with area $(\frac{2}{3}n)^2+O(n)$ \cite{Bie-SWAT96},
while the best-possible area for planar 2-connected graphs is $n^2$ \cite{TTV91}.

\section{Preliminaries}

Let $G=(V,E)$ be a graph with $n=|V|$ vertices and $m=|E|$ edges.
The {\em degree} of a vertex $v$ is the number of incident edges.
In this paper all graphs are assumed to be {\em 4-graphs}, i.e.,
all vertex degrees are at most 4.  A graph is called {\em 4-regular}
if every vertex has degree exactly 4; such a graph has $m=2n$ edges.

A graph $G$ is called {\em connected} if for any two vertices $u,v$ there
is a path in $G$ connecting $u$ and $v$.  It is called 
3-connected if $n>3$ and for any two vertices $u,v$ the graph $G-\{u,v\}$
is connected.

A {\em loop} is an edge $(v,v)$ that connects an endpoint with
itself.  A {\em multi-edge} is an edge $(u,v)$ for which another
copy of edge $(u,v)$ exists. When not otherwise stated the
graph $G$ that we want to draw is {\em simple}, i.e., it has neither loops 
nor multi-edges.  While modifying $G$, we will sometimes temporarily add
a {\em double edge}, i.e., an edge for which exactly one other 
copy exists (we refer always to the added edge as double edge, the copy is not a double edge).

\subsection{The 3-canonical order}

\begin{definition}
Let $G$ be a 3-connected graph. A {\em 3-canonical order} is a partition of $V$ into
groups $V=V_1\cup \dots \cup V_k$ such that
\begin{itemize}
\item $V_1=\{v_1,v_2\}$, where $(v_1,v_2)$ is an edge.
\item $V_k=\{v_n\}$, where $(v_1,v_n)$ is an edge.
\item For any $1<i<k$, one of the following holds:
	\begin{itemize}
	\item $V_i=\{z\}$, where $z$ has at least two predecessors and at least one successor.  
	\item $V_i=\{z_1,\dots,z_\ell\}$ for some $\ell\geq 2$, where 
		\begin{itemize}
		\item $z_1,\dots,z_\ell$ is an induced path in $G$ (i.e. edges
			$z_1-z_2-\dots-z_\ell$ exist, and there are no
			edges $(z_i,z_j)$ with $i<j-1$),
		\item $z_1$ and $z_\ell$ have exactly one predecessor each, 
			and these predecessors are different,
		\item $z_j$ for $1<j<\ell$ has no predecessor,
		\item $z_j\in V_i$ for $1\leq j\leq \ell$ has at least one successor.
		\end{itemize}
	\end{itemize}
\end{itemize}

	Here, a {\em predecessor} [{\em successor}] of a vertex in $V_i$ is a neighbor that occurs in a group $V_h$ with 
	$h<i$ [$h>i$].  

\end{definition}

We call a vertex group $V_i$ a \emph{singleton} if $|V_i|=1$, and a {\em chain} if $|V_i|\geq 2$ and $i\geq 2$.
Distinguish chains further into {\em short chains} with $|V_i|=2$ and {\em long chains} with $|V_i|\geq 3$.
A 3-canonical order imposes a natural orientation on the edges of the graph from lower-indexed groups to
higher-indexed groups and, for edges within a chain-group, from one (arbitrary) end of the path to the other. This implies $\indeg(v)\geq 2$ for any singleton, $\indeg(v)=2$ for exactly one vertex of each chain, and
$\indeg(v)=1$ for all other vertices of a chain.

\begin{figure}[ht]
\hspace*{\fill}
\includegraphics[width=55mm,page=1]{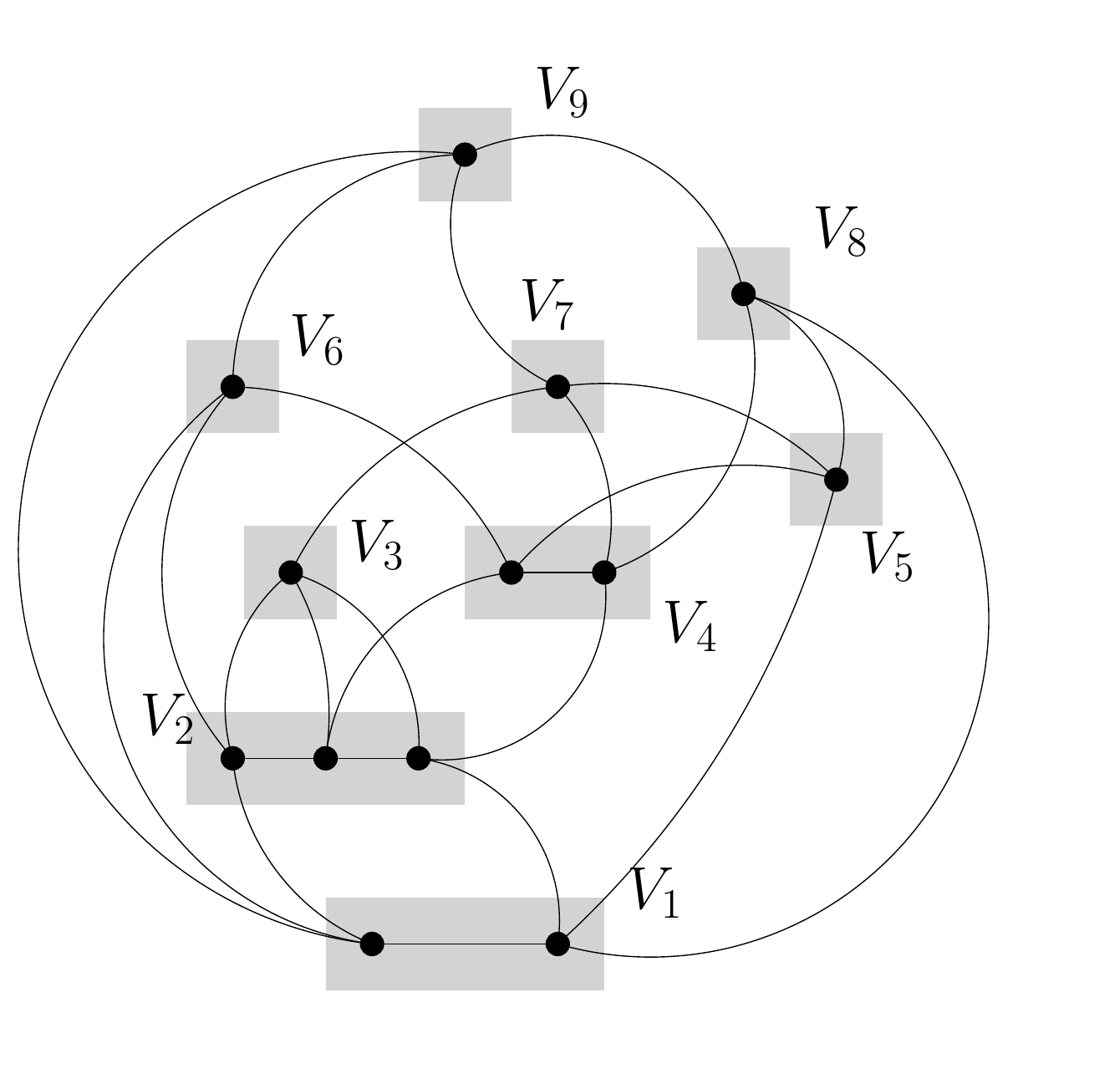}
\hspace*{\fill}
\includegraphics[width=55mm,page=2]{example.pdf}
\hspace*{\fill}
\caption{A 4-regular 3-connected graph with a 3-canonical order,
and the drawing created with our algorithm\protect\footnotemark.
$V_2$ is a long chain, $V_4$ is a short chain,
$V_5$ is a 2-2-singleton, $V_3,V_6,V_7$ and $V_8$ are 3-1-singletons.
}
\label{fig:example}
\end{figure}
\footnotetext{For
illustrative purposes we
show the drawing exactly as created, even though many more grid lines and
bends could be saved with straightforward compaction steps.}

Numerous related methods of ordering vertices of 3-connected graphs exist, e.g.\ the \emph{(2,1)-sequence}~\cite{Mondshein1971}, 
the \emph{non-separating ear decomposition}~\cite{Cheriyan1988,Schmidt2014}, and, limited to planar
graphs, the {\em canonical order for maximal planar graphs}~\cite{FPP90}, the
{\em canonical order for 3-connected planar graphs}~\cite{Kant1996} and \emph{orderly spanning trees}~\cite{Chiang2005}.
A 3-canonical order of a 3-connected graph implies all these orders, up to minor subtleties.

A convenient way to prove that a 3-canonical order exists and can be computed efficiently is to use non-separating ear decompositions. This is a partition of the edges into {\em ears} $P_1\cup \dots \cup P_k=E$ such that
$P_1$ is an induced cycle, $P_i$ for $i>1$ is a non-empty induced path (possibly consisting of one edge) that intersects $P_1\cup \dots \cup P_{i-1}$ in 
exactly its endpoints,
and $G-\left( \bigcup_{h=1}^{i} P_h\right)$ is connected for every $i < k$. 
Such a non-separating ear decomposition exists for any 3-connected graph~\cite{Cheriyan1988}, and we can
even fix the edge $(v_1,v_n)$ and require that $v_1$ is in the cycle $P_1$ and that $v_n$ is the only vertex in $P_k$; hence $P_k$ will be a singleton.
Further, such a non-separating ear decomposition (under the name \emph{Mondshein sequence}) 
can be computed in linear time~\cite{Schmidt2014}. The sets of newly added vertices in $P_i$ for $i=1,\dots,k$,
omitting empty groups, will be the vertex groups of a 3-canonical order. Although the ears in such a decomposition allow vertices in a chain $V_i$ to have arbitrarily many incoming incident edges, we can get rid of these extra edges by short-cutting ears (see Lemmas 8 and 12 in~\cite{Schmidt2013}). Therefore one can obtain a 3-canonical order from a Mondshein sequence in linear time.

\subsection{Making 4-graphs 4-regular}

It will greatly simplify the description of the algorithm if we only give it
for 4-regular graphs.  Thus, we want to modify a 4-graph $G$ 
such that the resulting graph $G'$ is 4-regular, draw $G'$,
and then delete added edges to obtain a drawing of $G$.
However, we must maintain a simple graph since the existence
of 3-canonical orders depends on simplicity.  This turns out
to be impossible, but allowing one double edge is sufficient.

\begin{lemma}
\label{lem:4regular}
Let $G$ be a simple 3-connected 4-graph with $n\geq 5$.
Then we can add edges to $G'$ such that the resulting
graph $G'$ is 3-connected, 4-regular, and has at most
one double edge.
\end{lemma}
\begin{proof}
Since $G$ is 3-connected, any vertex has degree 3 or 4.
If there are four or more vertices of degree 3, then 
they cannot be mutually adjacent (otherwise $G=K_4$,
which contradicts $n\geq 5$).  So then we can add an edge 
between two non-adjacent vertices of degree 3; this 
maintains simplicity and 3-connectivity.

Repeat until only two vertices of degree 3 are left (recall that the number of vertices of odd degree is even).
Now add an edge between these two vertices even if there
existed one already; this edge is the only one that may
become a double edge.  The resulting graph is 4-regular
and satisfies all conditions.
\end{proof}

\section{Creating orthogonal drawings}
From now on let $G$ be a 3-connected 4-regular graph that has no loops and
at most one double edge.  Compute a 3-canonical order
$V=V_1\cup \dots \cup V_k$ of $G$ with $V_k=\{v_n\}$, choosing $v_1v_n$ to be the double edge if there is one.
Let $x^{\text{short}}$ and $x^{\text{long}}$ be the number of short and long chains.
Let $\x{j}{\ell}$ be the number of  vertices with in-degree $j$ and out-degree $\ell$.
Since $G$ is 4-regular, we must have $j+\ell=4$.
A {\em $j$-$\ell$-singleton} is a vertex $z$ that constitutes a
singleton group $V_i$ for $1 < i \leq k$ and that has in-degree $j$ and out-degree $\ell$.
Using properties of the 3-canonical order and some edge-counting arguments,
the following is easily shown:

\begin{observation}
\label{obs:xij}
Let $G$ be a 4-regular graph with a 3-canonical order.   Then
\begin{enumerate}
\item $\x{0}{4}=\x{4}{0}=1$.
\item $\x{1}{3}=\x{3}{1}+\Theta(1)$.
\item Every chain $V_i$ contributes one to $\x{2}{2}$ and $|V_i|-1$
	to $\x{1}{3}$.
\end{enumerate}
\end{observation}

\subsection{A simple algorithm}

As in many previous orthogonal drawing papers~\cite{Biedl1998,Kant1996,Papakostas1998} the idea is to draw the graph $G_i$ induced by
$V_1\cup \dots \cup V_i$ in such a way that all {\em unfinished edges}
(edges with one end in $G_i$ and the other in $G-G_i$) end in a column
that is unused above the point where the drawing ends.

\paragraph{Embedding the first two vertices: } If $(v_1,v_n)$ is a single edge, then $v_1$ and $v_2$ are embedded exactly as in~\cite{Biedl1998}; refer to Fig.~\ref{fig:embedBK}.
If $(v_1,v_n)$ is a double edge, then it was added only for the purpose of making the graph 4-regular and need not be drawn. In that case we omit one of the outgoing edges of $v_1$ having a bend.

\paragraph{Embedding a singleton: } If $V_i$ is a singleton $\{z\}$, we embed $z$ exactly as in~\cite{Biedl1998}; refer to Fig.~\ref{fig:embedBK}.
For $\indeg(z)\in \{2,3\}$, this adds one new row and $\outdeg(z)-1=3-\indeg(z)$ many new columns. For $\indeg(z) = 4$, $z = v_n$; if $(v_1,v_n)$ is a double edge, we omit the edge having two bends.

\begin{figure}[b]
\hspace*{\fill}
\includegraphics[page=4,width=27mm]{BK_cases.pdf}
\hspace*{\fill}
\includegraphics[page=1,width=27mm]{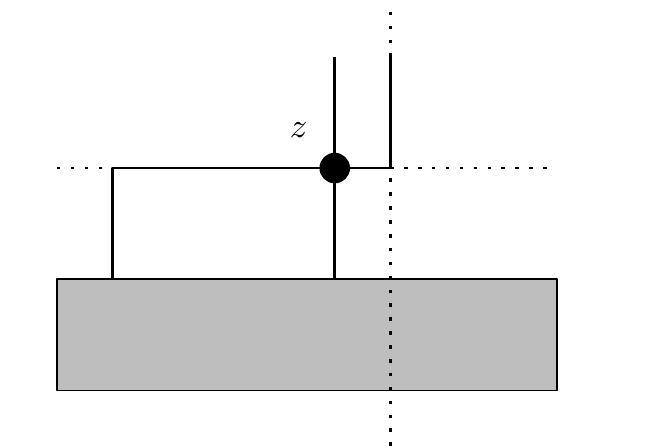}
\hspace*{\fill}
\includegraphics[page=2,width=27mm]{BK_cases.pdf}
\hspace*{\fill}
\includegraphics[page=3,width=27mm]{BK_cases.pdf}
\hspace*{\fill}
\caption{Embedding the first two vertices, and a singleton with
in-degree $2,3,4$.  Newly added grid-lines are dotted.}
\label{fig:embedBK}
\end{figure}

\paragraph{Embedding chains: } Let $V_i$ be a chain,
say $V_i=\{z_1,\dots,z_\ell\}$ with $\ell\geq 2$.  For chains, our algorithm is
substantially different from~\cite{Biedl1998}.  Only $z_1$ and $z_\ell$
have predecessors.  We place the chain-vertices on a new horizontal row
above the previous drawing, between the edges from the predecessors; see Fig.~\ref{fig:EmbeddingChains}.
We add new columns as needed to have space for new vertices and outgoing
edges without using columns that are in use for other unfinished edges.
We also use a second new row if the chain is a long chain. 

\begin{figure}[ht]
\hspace*{\fill}
\includegraphics[page=2,width=50mm]{standardEmbeddingNew.pdf}
\hspace*{\fill}
\includegraphics[page=1,width=50mm]{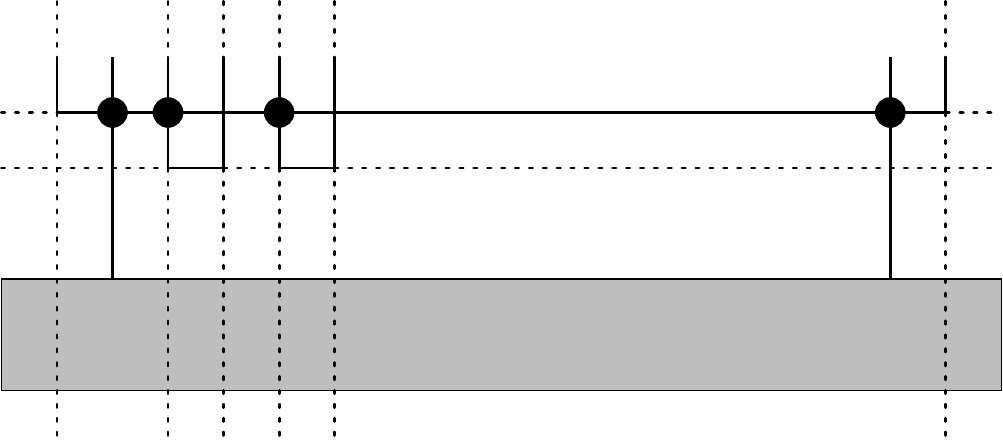}
\hspace*{\fill}
\caption{Embedding short and long chains.} 
\label{fig:EmbeddingChains}
\end{figure}

\begin{observation}
\label{obs:increase}
The increase in the half-perimeter is as follows:
\begin{itemize}
\item For the first and last vertex-group: $O(1)$
\item For a 3-1-singleton: $+1$ (we add one row)
\item For a 2-2-singleton: $+2$ (we add one row and one column)
\item For a short chain: $+3$ (we add one row and two columns)
\item For a long chain $V_i$: $+2|V_i|$ (we add two rows and $2|V_i|-2$ columns)
\end{itemize}
\end{observation}
\begin{corollary}
\label{cor:HPsimple}
The half-perimeter is at most $ \frac{3}{2}n + \frac{1}{2}\x{2}{2} 
- x^\text{short} + O(1)$.
\end{corollary}
\begin{proof}
From Observation~\ref{obs:increase} and using Observation~\ref{obs:xij}.3
the half-perimeter is at most
$\x{3}{1}+2\x{2}{2}+2\x{1}{3}-x^\text{short}+O(1)$.
By Observation~\ref{obs:xij}.2 this is
at most $\frac{3}{2}\x{3}{1}+2\x{2}{2}+\frac{3}{2}\x{1}{3}-x^\text{short}+O(1)$,
which gives the result.
\end{proof}

\begin{theorem}
Every simple 3-connected 4-graph has an orthogonal drawing of area at most 
$\frac{25}{36}n^2+O(n)\approx 0.69n^2$.
\end{theorem}
\begin{proof}
In a nutshell, the above algorithm or the one in \cite{Papakostas1998}
gives the desired bound.  In more detail,
make the graph 4-regular, compute the 3-canonical order, and consider
the number $\x{2}{2}$ of 2-2-vertices.
\begin{enumerate}
\item If $\x{2}{2}\leq n/3$, then apply the above algorithm.  By
	Corollary~\ref{cor:HPsimple}, the half-perimeter is at 
	most $\frac{3}{2}n+\frac{1}{6}n+O(1) \leq \frac{5}{3}n + O(1)$.
\item If $\x{2}{2}\geq n/3$, apply the algorithm from~\cite{Papakostas1998}.
	They state their area bound as $0.76n^2+O(1)$, but 
	in fact their half-perimeter is at most
	$2n-\frac{1}{2}(\x{1}{3}+\x{2}{2}) + O(1)$.
	Using Observation~\ref{obs:xij}.2 and ignoring $O(1)$ terms, we have $\x{1}{3}+\x{2}{2} = \frac{1}{2}\x{1}{3} + \x{2}{2} + \frac{1}{2}\x{3}{1} = \frac{1}{2}n+\frac{1}{2}\x{2}{2}$.
	Hence the half-perimeter of their algorithm is at most $\frac{7}{4}n-\frac{1}{4}\x{2}{2}+O(1)
	\leq (\frac{7}{4}-\frac{1}{12}) n +O(1) = \frac{5}{3}n+O(1)$.
\end{enumerate}
So either way, we get a drawing with half-perimeter $\frac{5}{3}n + O(1)$.
The area of it is maximal if the two sides are equally large and thus at most $(\frac{5}{6}n+O(1))^2$.
\end{proof}

\subsection{Improvement via pairing}

We already know a bound of $\frac{3}{2}n+\frac{1}{2}\x{2}{2} - x^\text{short}+O(1)$
on the half-perimeter. This section improves this further to half-perimeter
$\frac{3}{2}n+O(1)$. 
The idea is strongly inspired by the pairing technique of Papakostas and Tollis~\cite{Papakostas1998}. They created pairs of vertices with
special properties such that at least $\frac{1}{2}(\x{2}{2}+\x{1}{3})$ 
such pairs must exist. For each pair they
can save at least one grid-line, compared to the $2n+O(1)$ grid-lines
created with~\cite{Biedl1998}.

Our approach is similar, but instead of pairing vertices, we pair groups of
the canonical order
by scanning them in backward order as follows:
\begin{enumerate}
\item Initialize $i:=k-1$.  (We ignore the last group, which is a 4-0-singleton.)
\item While $V_i$ is a 3-1-singleton and $i>2$, set $i:=i-1$.
\item If $i=2$, break.  Else $V_i$ is a chain or a 2-2-singleton and
	we choose the partner of $V_i$ as follows:
Initialize $j:=i-1$.
While $V_j$ is a 3-1-singleton whose successor is not in $V_i$, set $j:=j-1$.
Now pair $V_i$ with $V_j$.  
	Observe that such a $V_j$ with $j\geq 2$ always exists since $i>2$ and $V_2$
	is not a 3-1-singleton.
\item Update $i=j-1$ and repeat from Step (2) onwards.
\end{enumerate}

In the small example in Fig.~\ref{fig:example}, the 2-2-singleton $V_5$
gets paired with the short chain $V_4$, and all other groups are not paired.

Observe that with the possible exception of $V_2$, every chain is paired
and every 2-2-vertex is in a paired group (either as 2-2-singleton or as part
of a chain). Hence, there are at least $\frac{1}{2}(\x{2}{2}-1)$ pairs. 
The key observation is the following:

\begin{lemma}
\label{lem:grouping_works}
Let $V_i,V_j$ be two vertex groups that are paired. Then there exists
a method of drawing $V_i$ and $V_j$ (without affecting the layout of
any other vertices) such that the increase to rows and columns is
at most $2|V_i\cup V_j|-1$.
\end{lemma}

We defer the (lengthy) proof of Lemma~\ref{lem:grouping_works} 
to the next section, and study here first its consequences.
We can draw $V_1$ and $V_k$ using $O(1)$ grid-lines.  We can draw $V_2$
using $2|V_2|=2\x{2}{2}_{V_2}+2\x{1}{3}_{V_2}$ new grid-lines, 
where $\x{\ell}{k}_W$ denotes the number of vertices of in-degree $\ell$
and out-degree $k$ in vertex set $W$.  We can draw
any unpaired 3-1-singleton using one new grid-line. Finally,
we can draw each pair using
$ 2|V_i\cup V_j|-1=2\x{2}{2}_{V_i\cup V_j} + 2\x{1}{3}_{V_i\cup V_j} - 1$ new grid-lines.
This covers all vertices, since all 2-2-singletons and all chains 
belong to pairs or are $V_2$, and since there are no 1-3-singletons.

Putting it all together and using Observation~\ref{obs:xij}.2, 
the number of grid-lines hence is
$2\x{1}{3}+2\x{2}{2}+\x{3}{1}-\#\mbox{pairs}+O(1)
\leq 2\x{1}{3}+\frac{3}{2}\x{2}{2}+\x{3}{1}+O(1)
= \frac{3}{2}n + O(1)$ as desired.
Since a drawing with half-perimeter $\frac{3}{2}n$ has
area at most $(\frac{3}{4}n)^2=\frac{9}{16}n^2$, we can conclude:

\begin{theorem}
\label{thm:main}
Every simple 3-connected 4-graph has an orthogonal drawing of area at most 
$\frac{9}{16}n^2+O(n)\approx 0.56n^2$. 
\end{theorem}

We briefly discuss the run-time.  The 3-canonical
order can be found in linear time.  Most steps
of the drawing algorithm work in constant time per vertex, hence
$O(n)$ time total.  One difficulty is that to place a group we must
know the relative order of the columns of the edges from the
predecessors.  As discussed extensively in \cite{Biedl1998}, we
can do this either by storing columns as a balanced binary search
tree (which uses $O(\log n)$ time per vertex-addition), or using
the data structure by Dietz and Sleator \cite{DS87} which allows
to find the order in $O(1)$ time per vertex-addition. Thus,
the worst-case run-time to find the drawing is $O(n)$.

\section{Proof of Lemma~\protect\ref{lem:grouping_works}}

Recall that we must show that two paired vertex groups $V_j$ and $V_i$,
with $i<j$, can be embedded such that we use at most
$2|V_i|+2|V_j|-1$ new grid-lines.  The proof of this is a
massive case analysis, depending on which type of group
$V_i$ and $V_j$ are, and whether there are edges between them or not.%
\footnote{The constructions we give have been designed as to keep the
description simple; often even more grid-lines could be saved by doing
more complicated constructions.%
}  We first observe some properties of pairs.

\begin{observation} 
\label{obs:pairing}
By choice of the pairing, the following holds:
\begin{enumerate}
\item For any pair $(V_i,V_j)$ for $j<i$, group $V_i$ is either a 2-2-singleton
	or a chain.
\item If $V_i$ is paired with $V_j$ for $j<i$, then all predecessors of $V_i$ are in $V_j$ or occurred in a group before $V_j$.
\end{enumerate}
\end{observation}

The following notation will cut down the number of cases a bit.
We say that groups $V_i$ and $V_j$ are {\em adjacent} if there
is an edge from a vertex in one to a vertex in the other.
If two paired groups $V_i,V_j$ are not adjacent, then by Observation~\ref{obs:pairing}.2
all predecessors of $V_i$ occur before group $V_j$.  We hence can safely draw
$V_i$ first, and then draw $V_j$, thereby effectively exchanging the roles
of $V_i$ and $V_j$ in the pair.  Now we distinguish cases:

\begin{enumerate}
\item One of $V_i$ and $V_j$ is a short chain.  
Say $V_i$ is the short chain,
the other case is similar.  Recall that the standard layout for
a short chain uses 3 new grid-lines, but $\x{2}{2}_{V_i}+\x{1}{3}_{V_i}=2$.
So the layout of a short chain automatically saves one grid-line.
We do not change the algorithm at all in this case; laying out $V_i$
and $V_j$ exactly as before results in at most
$ 2\x{2}{2}_{V_i\cup V_j} + 2\x{1}{3}_{V_i\cup V_j} - 1$ new grid-lines.
(This is what happens in the example of Fig.~\ref{fig:example}.)

\item One of $V_i$ and $V_j$ is a 3-1-singleton.  
By Observation~\ref{obs:pairing}, the 3-1-singleton must be $V_j$.
By the pairing algorithm, the unique outgoing edge of the 3-1-singleton
must lead to $V_i$.  Draw $V_j$ as before.
We can then draw $V_i$ so that it re-uses one of the columns that were
freed by $V_j$. See Fig.~\ref{fig:3-1}.

\begin{figure}[ht]
\hspace*{\fill}
\includegraphics[page=1,height=20mm]{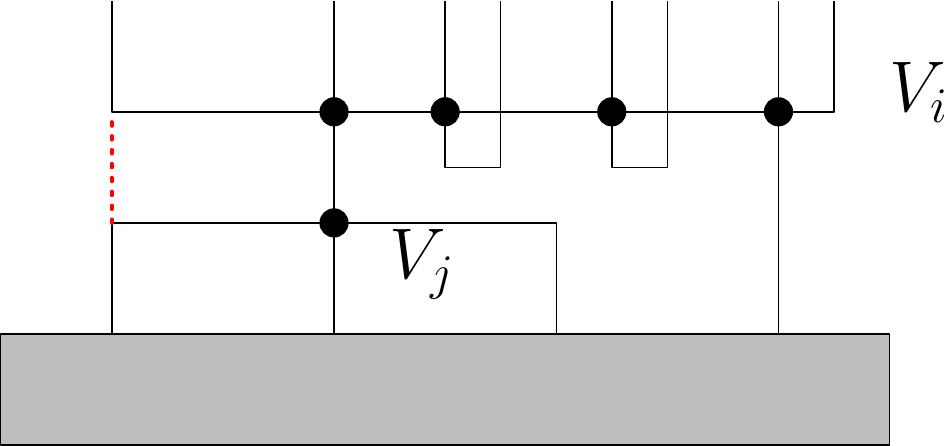}
\hspace*{\fill}
\includegraphics[page=3,height=20mm]{3-1.pdf}
\hspace*{\fill}
\caption{Reusing the column freed by a 3-1-singleton with a later
chain or singleton.  In this and the following figures, the
re-used grid-line is dotted.}
\label{fig:3-1}
\end{figure}

\item $V_i$ and $V_j$ are both long chains.
In this case, both $V_i$ and $V_j$ can use the same extra row
for the ``detours'' that their {\em middle} vertices 
(by which we mean  vertices that are neither the first nor the last vertex
of the chain)
use.  Since we can freely choose into which columns these
middle vertices are placed, we can ensure that none of these
``detours'' overlap and hence one row suffices for both chains.
This holds even if one or both of the predecessors of $V_i$ are 
in $V_j$, as these are distinct and the two corresponding incoming edges of $V_i$ 
extend the edges that were already drawn for $V_j$.
See Fig.~\ref{fig:long-chains-independent}.

\begin{figure}[ht]
\hspace*{\fill}
\includegraphics[page=2,height=20mm]{long-chains-independent.pdf}
\hspace*{\fill}
\includegraphics[page=1,height=20mm]{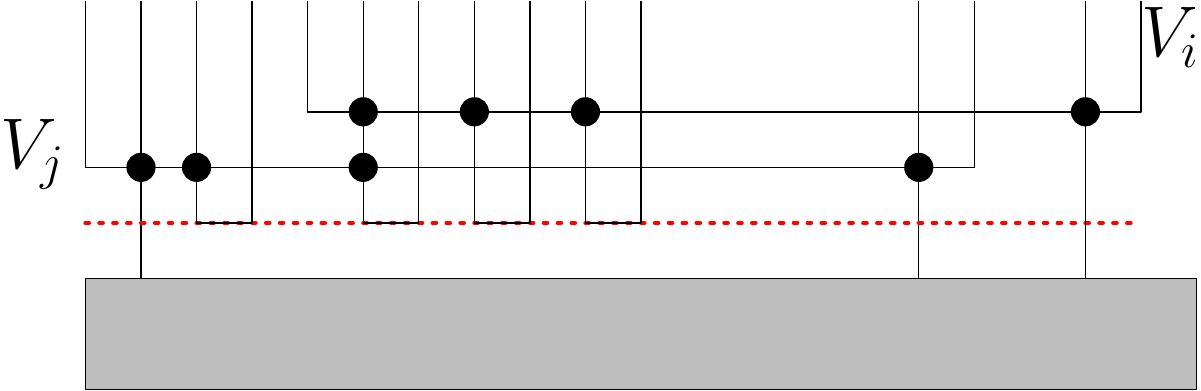}
\hspace*{\fill}
\\[2ex]
\hspace*{\fill}
\includegraphics[page=3,height=20mm]{long-chains-independent.pdf}
\hspace*{\fill}
\caption{Sharing the extra row between two long chains when there are 0, 1 or 2 predecessors of $V_i$ in $V_j$.} 
\label{fig:long-chains-independent}
\end{figure}

\item None of the previous cases applies and $V_j$ is a 2-2-singleton.  
By Observation~\ref{obs:pairing}.1 and since Case~(1) does not apply, $V_i$ is either a 2-2-singleton or a long chain.
There are two columns reserved for edges from predecessors of $V_j$.
Since predecessors of $V_i$ are distinct, at most one of them can be
the 2-2-singleton in $V_j$.  So there also is at least one column
reserved for an edge from a predecessor of $V_i$ not in $V_j$.
We call these three or four columns the {\em predecessor-columns}. 
We have sub-cases depending on the relative location of these columns:

\begin{enumerate}
\item The leftmost predecessor-column leads to $V_j$.
In this case we save a column almost exactly as in~\cite{Papakostas1998}.  
Place $V_j$ as before, in the right one of its predecessor-columns.
This leaves the leftmost predecessor-column free to be reused.  Now
no matter whether $V_i$ is a 2-2-singleton or a long chain, or whether
$V_i$ is adjacent to $V_j$ or not, 
we can re-use this leftmost column for one outgoing edge of $V_i$ with a 
suitable placement. See Fig.~\ref{fig:2-2-single-independent}.

\begin{figure}[ht]
\hspace*{\fill}
\includegraphics[page=2,height=20mm]{2-2-Vj-leftmost.pdf}
\hspace*{\fill}
\includegraphics[page=1,height=20mm]{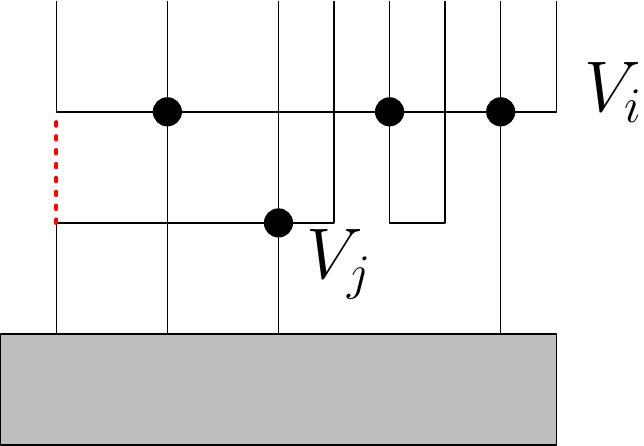}
\hspace*{\fill}
\includegraphics[page=4,height=20mm]{2-2-Vj-leftmost.pdf}
\hspace*{\fill}
\includegraphics[page=3,height=20mm]{2-2-Vj-leftmost.pdf}
\hspace*{-2mm}
\caption{Reusing the leftmost predecessor-column freed by a 2-2-singleton $V_j$ in Case~4(a). Left two pictures: $V_i$ not adjacent to $V_j$. Right two: $V_i$ adjacent to $V_j$.}
\label{fig:2-2-single-independent}
\end{figure}

\item The rightmost predecessor-column leads to $V_j$.  This case is symmetric to the previous one.

\item The leftmost and rightmost predecessor-columns lead to $V_i$.
	This implies that $V_i$ has two predecessors not in $V_j$,
	hence $V_i$ cannot be adjacent to $V_j$.  If $V_i$ is a
	2-2-singleton, then (as discussed earlier) we can exchange
	the roles of $V_i$ and $V_j$, which brings us to Case 4(a).
	If $V_i$ is a long chain, then 
place $V_j$ in the standard fashion. We then place the long
chain $V_i$ such that the ``detours'' of the middle vertices 
re-use the row of $V_j$.
See Fig.~\ref{fig:2-2-single-long-independent}.

\begin{figure}[ht]
\hspace*{\fill}
\includegraphics[page=4,height=20mm]{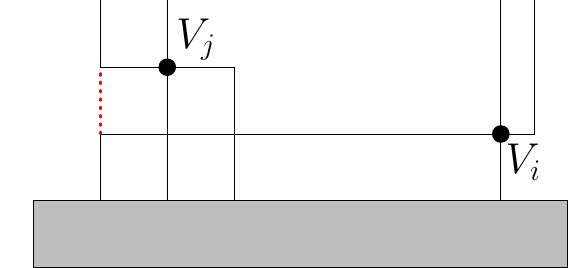}
\hspace*{\fill}
\includegraphics[page=3,height=20mm]{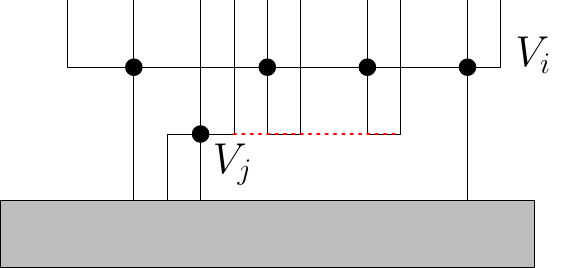}
\hspace*{\fill}
\caption{
If the predecessor-columns of $V_j$ are between the ones of $V_i$,
then we can either revert to Case 4(a) or the long chain $V_i$ can re-use the row of $V_j$.}
\label{fig:2-2-single-long-independent}
\end{figure}
\end{enumerate}

\item None of the previous cases applies and $V_j$ is a chain.
Say $V_j=\{z_1,\dots,z_\ell\}$, where $\ell\geq 3$ since
Case (1) does not apply.  We assume the naming is such that
the predecessor column of $z_1$ is left of the predecessor column
of $z_\ell$.

Since we are not in a previous case, $V_i$ must be a 2-2-singleton, say $z$.
If $V_i$ is not adjacent to
$V_j$, then we can again exchange the roles of $V_i$ and $V_j$
and are in Case~(4).
So we may assume that there are edges between $V_j$
and $V_i$.  We distinguish the following sub-cases depending on how
many such edges there are and whether their ends are middle vertices.

\begin{enumerate}
\item $z$ has exactly one neighbor in $V_j$, and it is either $z_1$ or $z_\ell$.
We rearrange $V_i\cup V_j$ into two different chains.
Let $z$ be adjacent to $z_1$ (the other case is symmetric). Then $\{z,z_1\}$ forms one chain
and $\{z_2,\dots,z_\ell\}$ forms another.  Embed these two chains
as usual.  Since $\{z,z_1\}$ forms a short chain, this saves one grid-line. 
See Fig.~\ref{fig:2-2-single-long-edges}(left).

\item $z$ has exactly one neighbor in $V_j$, and it is $z_h$ for some $1<h<\ell$.
Embed the chain $V_j$ as usual, but omit the new column
next to $z_h$. For embedding $z$, we place a new row {\em below} the rows for the chain; 
using this new row we can connect the bottom outgoing edge of $z_h$
to the horizontal incoming edge of $z$.
See Fig.~\ref{fig:2-2-single-long-edges}(right).

\begin{figure}[ht]
\hspace*{\fill}
\includegraphics[page=1,height=20mm]{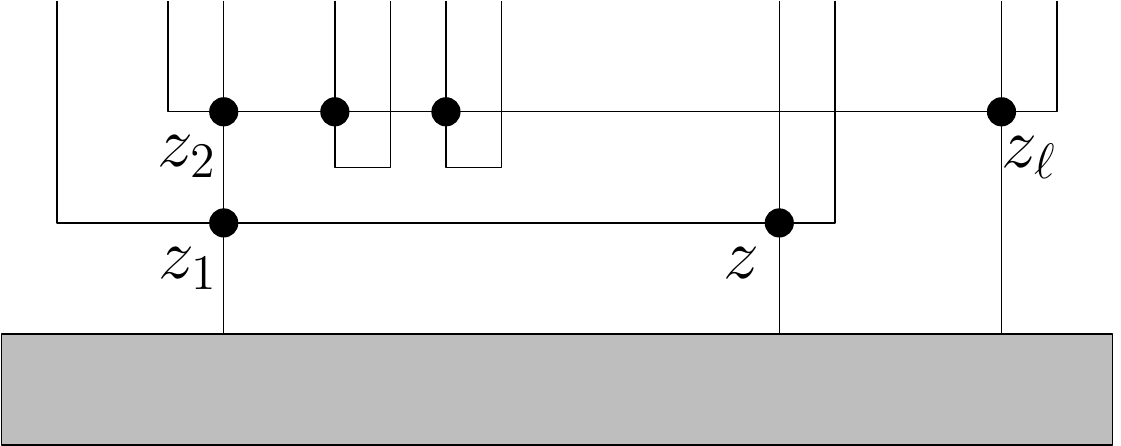}
\hspace*{\fill}
\includegraphics[page=2,height=20mm]{2-2-single-long-edges.pdf}
\hspace*{\fill}
\caption{$V_j$ is a long chain, $V_i$ is a 2-2-singleton with
one predecessor in $V_j$.}
\label{fig:2-2-single-long-edges}
\end{figure}

\item $z$ has two neighbors in $V_j$, and both of them are middle
	vertices $z_{g},z_{h}$ for $1<g<h<\ell$.
Embed the chain $V_j$ as usual, but omit the new columns
next to $z_g$ and $z_h$.  Place a new row {\em between} the two rows for 
the chain and use it to connect the two bottom outgoing edges 
of $z_g$ and $z_h$ to place
$z$, re-using the row for the detours
to place the bottom outgoing edge of $z$.  This uses
an extra column for $z$, but 
saved two columns at $z_g$ and $z_h$, so overall one grid-line has been
saved. See Fig.~\ref{fig:2-2-single-long-edges2}(top left).

\item $z$ is adjacent to $z_1$ and $z_2$ (the case of adjacency to
	$z_{\ell-1}$ and $z_\ell$ is symmetric).
Embed $z_2,\dots,z_\ell$ as usual for a chain, then place $z_1$ below $z_2$.
The horizontally outgoing edge of $z_2$ intersects one outgoing
edge of $z_1$; put $z$ at this place to save a row and a column. See Fig.~\ref{fig:2-2-single-long-edges2}(top right).

\item $z$ is adjacent to $z_1$ and $z_h$ with $h>2$ (the
	case of adjacency to $z_\ell$ and $z_h$ with $h<\ell-1$ is symmetric).
Draw the chain $V_j$ with the modification that $z_h$ is {\em below}
$z_{h-1}$, but still all middle vertices use the same extra row for
their downward outgoing edges.  This uses 3 rows, but now $z$ can
be placed using the two left outgoing edges of $z_1$ and $z_h$,
saving a row for $z$ and a column for the left outgoing edge of $z_h$.
See Fig.~\ref{fig:2-2-single-long-edges2-more}(bottom), both for $h<\ell$ and $h=\ell$.
\end{enumerate}
\end{enumerate}
This ends the proof of Lemma~\ref{lem:grouping_works} and hence
shows Theorem~\ref{thm:main}.

\begin{figure}[ht]
\hspace*{\fill}
\includegraphics[page=5,height=20mm]{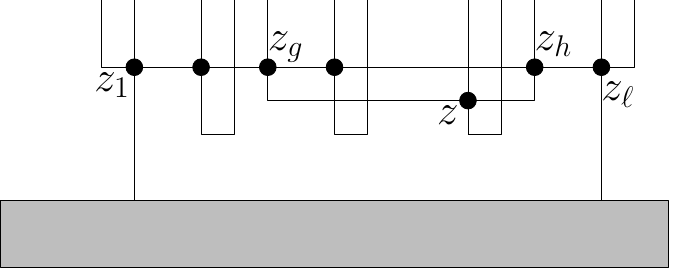}
\hspace*{\fill}
\includegraphics[page=6,height=20mm]{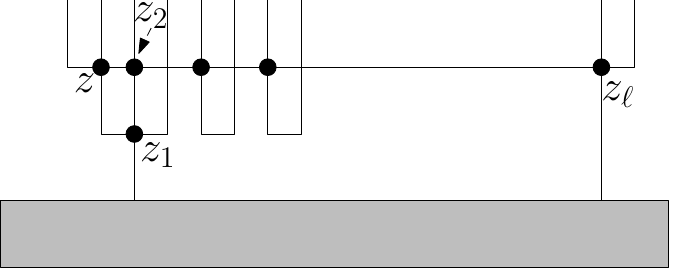}
\hspace*{\fill}
\\[2ex]
\hspace*{\fill}
\includegraphics[page=3,height=20mm]{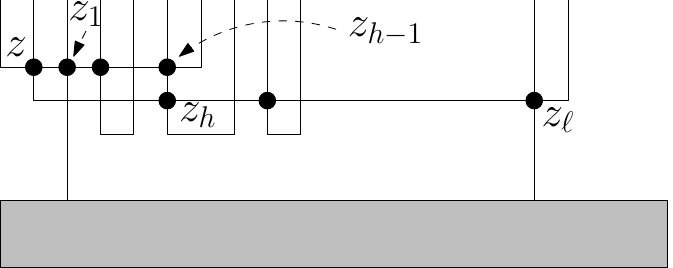}
\hspace*{\fill}
\includegraphics[page=4,height=20mm]{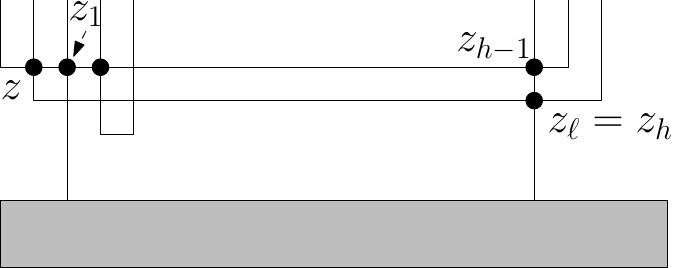}
\hspace*{\fill}
\caption{$V_j$ is a long chain, $V_i$ is a 2-2-singleton,
and there are exactly two edges between them.  (Top) Cases 5(c) and (d).
(Bottom) Case 5(e).}
\label{fig:2-2-single-long-edges2}
\label{fig:2-2-single-long-edges2-more}
\end{figure}

\section{Conclusion}

In this paper, we gave an algorithm to create an orthogonal drawing
of a 3-connected 4-graph that has area at most
$\frac{9}{16}n^2+O(n)\approx 0.56n^2$. As a main tool we used the 3-canonical order for arbitrary 3-connected graphs, whose existence 
was long known but only recently made efficient. To our knowledge, this is the first application of the 3-canonical
order on non-planar graphs in graph-drawing.
%
Among the many remaining open problems are the following:
\begin{itemize}
\item Can we draw 2-connected graphs with area less than
	$0.76n^2$?  A natural approach would be to draw each
	3-connected component with area $0.56n^2$ and to merge
	them suitably, but there are many cases depending on how
	the cut-vertices and virtual edges are drawn, and so this
	is far from trivial.
\item Can we draw 3-connected graphs with $(2-\varepsilon)n$ bends,
	for some $\varepsilon>0$?  With an entirely different algorithm
	(not given here), we have been able to prove a bound of
	$2n-\x{2}{2}+O(1)$ bends, so an improved bound seems likely.
\item Our algorithm was strongly inspired by the one of Kant \cite{Kant1996}
	for 3-connected planar graphs.  Are there other graph drawing algorithms 
	for planar 3-connec\-ted graphs
	that can be transferred to non-planar 3-connected graphs by using
	the 3-canonical order? 
\end{itemize}

\bibliographystyle{abbrv}
\bibliography{paper,journals,full,gd,papers}


\end{document}